\documentclass[11pt]{article}

\usepackage[T5]{fontenc}
\usepackage[margin=1in]{geometry}
\usepackage[labelfont=bf]{caption}
\usepackage{tabularray}
\usepackage{url}
\usepackage{algorithm,algpseudocode,amsmath,amssymb,amsthm,float,mathtools,thmtools,thm-restate,scalerel,stackengine,verbatim,subcaption,bbm}
\usepackage[shortlabels]{enumitem}
\usepackage{tikz}
\usetikzlibrary{decorations.pathreplacing, positioning}

\usepackage[pagebackref,colorlinks,citecolor=blue,,linkcolor=magenta,bookmarks=true]{hyperref}
\usepackage[nameinlink,capitalize]{cleveref}
 
\stackMath
\newcommand\reallywidehat[1]{%
\savestack{\tmpbox}{\stretchto{%
  \scaleto{%
      \scalerel*[\widthof{\ensuremath{#1}}]{\kern-.6pt\bigwedge\kern-.6pt}%
          {\rule[-\textheight/2]{1ex}{\textheight}}
            }{\textheight}%
            }{0.5ex}}%
            \stackon[1pt]{#1}{\tmpbox}%
            }
\newcommand{\Cbb}{\mathbb{C}}

\newcommand{\Nbb}{\mathbb{N}}

\newcommand{\wh}[1]{\widehat{#1}}
\newcommand{\eps}{\varepsilon}
\newcommand{\Tr}{\text{Tr}}

\newcommand{\expect}[2][]{\operatorname*{\mathbb{E}}_{#1 }\brac*{#2}}

\newcommand{\bO}[1]{\operatorname*{O}\paren*{#1}}

\newcommand{\bOm}[1]{\operatorname*{\Omega}\paren*{#1}}

\newcommand{\bT}[1]{\operatorname*{\Theta}\paren*{#1}}

\DeclarePairedDelimiter\norm{\lVert}{\rVert}
\DeclarePairedDelimiter\abs{|}{|}
\DeclarePairedDelimiter\brac{\lbrack}{\rbrack}
\DeclarePairedDelimiter\set{\lbrace}{\rbrace}
\DeclarePairedDelimiter\paren{\lparen}{\rparen}

\DeclareMathOperator{\tra}{Tr}

\DeclarePairedDelimiter\bra{\langle}{|}
\DeclarePairedDelimiter\ket{|}{\rangle}
\DeclarePairedDelimiterX\braket[2]{\langle}{\rangle}{#1\delimsize\vert#2}
\DeclarePairedDelimiterX\braHket[3]{\langle}{\rangle}{#1\delimsize\vert#2\delimsize\vert#3}

\newcommand{\maxentangled}{\sigma_{I^{\otimes n}}}

\newtheorem{theorem}{Theorem}
\newtheorem{lemma}[theorem]{Lemma}
\newtheorem{fact}[theorem]{Fact}
\newtheorem{corollary}[theorem]{Corollary}

\theoremstyle{definition}
\newtheorem{definition}[theorem]{Definition}
\theoremstyle{remark}
\newtheorem{remark}[theorem]{Remark}

\newcommand{\nth}{^\text{th}}

\newcommand{\pauli}{\mathcal{P}^{\otimes n}}

{\makeatletter
	\gdef\commentmark{%
		\expandafter\ifx\csname @mpargs\endcsname\relax 
		\expandafter\ifx\csname @captype\endcsname\relax 
		\marginpar{comment}
		\else
		comment 
		\fi
		\else
		comment 
		\fi}
	\gdef\comment{\@ifnextchar[\comment@lab\comment@nolab}
	\long\gdef\comment@lab[#1]#2{{\bf [\commentmark #2 ---{\sc #1}]}}
	\long\gdef\comment@nolab#1{{\bf [\commentmark #1]}}
}

\title{Hamiltonian Locality Testing via Trotterized Postselection}
\date{}
\author{John Kallaugher\\ Sandia National Laboratories\\ \texttt{jmkall@sandia.gov} \and Daniel Liang\\ Portland State University\\\texttt{danliang@pdx.edu}}

\begin{document}
\maketitle

\begin{abstract}
\noindent
The (tolerant) Hamiltonian locality testing problem, introduced
in~\lbrack Bluhm,~Caro,~Oufkir~`24\rbrack, is to determine whether a
Hamiltonian $H$ is $\eps_1$-close to being $k$-local (i.e.\ can be written as
the sum of weight-$k$ Pauli operators) or $\eps_2$-far from any $k$-local
Hamiltonian, given access to its time evolution operator and using as little total
evolution time as possible, with distance typically defined by the normalized
Frobenius norm.  We give the tightest known bounds for this problem, proving an
$\bO{\sqrt{\frac{\eps_2}{(\eps_2-\eps_1)^5}}}$ evolution time upper bound and an
$\bOm{\frac{1}{\eps_2-\eps_1}}$ lower bound.  Our algorithm does not require
reverse time evolution or controlled application of the time evolution
operator, although our lower bound applies to algorithms using either tool.

Furthermore, we show that if we \emph{are} allowed reverse time evolution, this
lower bound is tight, giving a matching $\bO{\frac{1}{\eps_2-\eps_1}}$ evolution time
algorithm.

\end{abstract}

\section{Introduction}

When dealing with large or
expensive-to-measure objects, learning the entire object may be too costly.
Property testing algorithms instead attempt to distinguish between the object
having a given property, or being far from any object with the property.
More generally, one can consider \emph{tolerant testing}, where one attempts to distinguish between the object
being within $\eps_1$-close to having a property, or being at least $\eps_2$-far from any object with the property.
Such algorithms have been extensively studied in quantum and classical settings
(see~\cite{MdW13} for an overview of the quantum case),
but~\cite{bluhm2024hamiltonianpropertytesting} was the first to consider it for
Hamiltonians accessed via their time evolution operator $e^{-iHt}$. In this
setting the natural measure of cost is \emph{total evolution
time}, $\sum_j
t_j$ where the $j\nth$ application of the time evolution operator is
$e^{-iH t_j}$.\footnote{Another cost measure that can be considered is total query
count, the number of individual applications of the time evolution operator. Our algorithm also uses the fewest number of
queries of any known algorithm.}

The property they considered was $k$-locality, a problem initially raised (but not studied) in \cite[Section 7]{MdW13} as well \cite{she_et_al:LIPIcs.ITCS.2023.96}. A Hamiltonian $H$ is $k$-local
if and only if it can be written as $\sum_j H_j$, where each $H_j$ operates on only $k$
qubits.
Such locality constraints (perhaps even geometrically locality constraints) are considered to be physically relevant.
Local Hamiltonians also appear to be theoretically relevant, as nearly all general learning algorithms for Hamiltonians assume that the Hamiltonian is local, whether they use the time evolution operator
\cite{huang2023many,haah2024learning,bakshi2024structure}, or copies of the Gibbs state \cite{anshu2021sample,bakshi2024learning}.
Local Hamiltonians are also conducive to efficient simulation on quantum computers, using the technique of Trotterization to break up the Hamiltonian into local quantum gate operations~\cite{lloyd1996universal}.
Finally, local Hamiltonians play an important role in quantum complexity theory, such as $\textsc{QMA}$-completeness and the Quantum PCP conjecture \cite{aharonov2013quantumpcpconjecture}.

The initial version of~\cite{bluhm2024hamiltonianpropertytesting} gave an
$\bO{n^{k+1}/(\eps_2)^3}$ evolution time algorithm when distance is
measured by the \emph{normalized} (divided by $2^{n/2}$ for a Hamiltonian
acting on $n$ qubits) Frobenius norm, improved
in~\cite{gutierrez2024simplealgorithmstestlearn} to $\bO{(\eps_2-\eps_1)^{-7}}$
and then in a later version of~\cite{bluhm2024hamiltonianpropertytesting} to
$\bO{(\eps_2-\eps_1)^{-2.5} \eps_2^{-0.5}}$.\footnote{The original
\cite{bluhm2024hamiltonianpropertytesting} algorithm only worked in the
intolerant setting of $\eps_1 =
0$.}\footnote{\cite{gutierrez2024simplealgorithmstestlearn} was later subsumed
by \cite{arunachalam2024testinglearningstructuredquantum}, which gives an
$\bO{(\eps_2-\eps_1)^{-3}}$ analysis.} This left open the question: how hard is
locality testing? Is it possible to achieve linear (a.k.a.\ Heisenberg) scaling
in $1/\eps$ for evolution time, and is such a scaling optimal in all error
regimes? In this work we make progress towards resolving the complexity of this
problem, improving the best known upper and lower bounds. Our algorithm is
based on a technique we refer to as \emph{Trotterized post-selection}, in which
we suppress the effect of local terms in the Hamiltonian evolution by
repeatedly evolving for a short time period and post-selecting on the non-local
part of the time evolution operator.

\subsection{Our Results}
Our main result is a improved upper bound for the Hamiltonian locality testing problem.
As with past works, our algorithm is also time-efficient and non-adaptive, though it does requires $n$ qubits of quantum memory, like \cite{gutierrez2024simplealgorithmstestlearn,arunachalam2024testinglearningstructuredquantum}.

\begin{restatable}{theorem}{upperfwdonly}
\label{thm:upperfwdonly}
Let $0 \le \eps_1 < \eps_2 \le 1$, $\delta \in (0, 1)$, and $k \in \Nbb$. There is an algorithm that distinguishes whether an $n$-qubit Hamiltonian $H$ is (1) within $\eps_1$ of some $k$-local Hamiltonian or (2) $\eps_2$-far
from all $k$-local Hamiltonians, with probability $1-\delta$.
The algorithm uses $\bO{\sqrt{\frac{\eps_2}{(\eps_2-\eps_1)^{7}}}\log(1/\delta)}$ non-adaptive queries to the time evolution operator with $\bO{\sqrt{\frac{\eps_2}{(\eps_2-\eps_1)^{5}}}\log(1/\delta)}$ total evolution
time.
\end{restatable}
We pair it with the first lower bound in the tolerant testing setting.
While our upper bound uses only forward time evolution and does not require controlled application of $e^{-itH}$, our
lower bound also applies to algorithms using either of these tools.
\begin{restatable}{theorem}{lowerbound}
\label{thm:lower}
Let $0 \le \eps_1 < \eps_2 \le 1$ and $k \in \Nbb$. Then any
algorithm that can distinguish whether an $n$-qubit Hamiltonian $H$ is (1) within $\eps_1$ of some $k$-local Hamiltonian or (2) $\eps_2$-far
from all $k$-local Hamiltonians,
must use $\bOm{\frac{1}{\eps_2 - \eps_1}}$ evolution time in
expectation to achieve constant success probability.
\end{restatable}

\begin{remark}
	\cite[Theorem 3.6]{bluhm2024hamiltonianpropertytesting} gives a hardness result for the \emph{unnormalized} Frobenius norm (as well as other Schatten norms) in the \emph{non-tolerant} setting that scales as $\Omega\left(\frac{2^{n/2}}{\eps}\right)$.
	Once normalized, this also gives a $\Omega\left(\frac{1}{\eps}\right)$ lower bound.
  However, this hardness result only holds for exponentially small $\eps$, due
  to the fact that the ``hard'' Hamiltonian in \cite[Lemma
  3.2]{bluhm2024hamiltonianpropertytesting} no longer has $\lVert H
  \rVert_\infty \leq 1$ when the \emph{unnormalized} Frobenius distance to
  $k$-local is super-constant.
	Therefore \cref{thm:lower} is, to the authors' knowledge, the first lower bound that works for arbitrary values of $\eps$, in addition to being the first for the tolerant setting. Our proof is also considerably simpler, and still extends to all of the distance measures considered in \cite{bluhm2024hamiltonianpropertytesting} and more.
\end{remark}

Finally, we show that, when reverse time evolution and controlled operations
are allowed, it is possible to saturate this lower bound even in the tolerant case.
\begin{restatable}{theorem}{ttinv}
\label{thm:tight-tolerant-inverse}
	Let $0 \le \eps_1 < \eps_2 \le 1$, $\delta \in (0, 1)$, and $k \in \Nbb$. There is an algorithm that tests whether an $n$-qubit Hamiltonian $H$ is (1) $\eps_1$-close to some $k$-local Hamiltonian or (2) $\eps_2$-far from all $k$-local Hamiltonians, with probability $1-\delta$.
	The algorithm uses $\bO{\frac{\log(1/\delta)}{(\eps_2-\eps_1)^2}}$ non-adaptive queries to the time evolution operator and its inverse, with $\bO{\frac{\log(1/\delta)}{\eps_2-\eps_1}}$ total evolution time.
\end{restatable}

\section{Proof Overview}
\subsection{Upper Bound}
For simplicity, we will consider the intolerant case ($\eps_1 = 0$, $\eps_2 =
\eps$) for this proof overview; the same techniques apply in the tolerant case
but require somewhat more care.
First we start with the intuition behind the algorithm of
\cite{gutierrez2024simplealgorithmstestlearn,arunachalam2024testinglearningstructuredquantum}.

We will need the fact that the space of $2n$ qubit states $\Cbb^{2^{2n}}$ has
the Bell basis $(\ket{\sigma_P})_P$, where $P$ spans the $n$-fold Paulis,
$\ket{\maxentangled}$ is the maximally entangled state
$\frac{1}{\sqrt{2^n}}\sum_{x \in \{0, 1\}^n} \ket{x}\otimes\ket{x}$, and
$\ket{\sigma_P} = (I^{\otimes n} \otimes P)\ket{\maxentangled}$.
Therefore, for any unitary $U$, if we apply $I^{\otimes n} \otimes U$ to
$\ket{\maxentangled}$ and then measure in the Bell basis, we are able to sample
from the (squared) Pauli spectrum\footnote{That is, $\alpha_P^2$ when $U$ is written as $\sum_P \alpha_p P$.} of $U$ (the squares of the Pauli
decomposition coefficients always sum to $1$ for a unitary
\cite{montanaro2010quantumbooleanfunctions}).

For any Hamiltonian $H$, the closest $k$-local Hamiltonian is given by dropping all of the non-local Paulis from its Pauli decomposition.
Therefore, as by the first-order Taylor series expansion, \[e^{-i H t} \approx I^{\otimes n} - i H t\]
for small enough $t$, we can set $U = e^{-i H \cdot t}$ in the
aforementioned procedure, and
if $H$ is $\eps$-far from local we will sample a non-local Pauli term with $\approx (t \cdot \eps)^2$ probability.
Conversely, if $H$ is local we should sample no non-local terms, giving us a distinguishing algorithm if the process is repeated $\bO{(t \cdot \eps)^{-2}}$ times, for a total time evolution of $\bO{t^{-1} \cdot \eps^{-2}}$.

So ideally we would like $t$ to be $\bT{1/\eps}$ and only repeat a constant
number of times, leading to a total time evolution of $\bO{\eps^{-1}}$, which
would be optimal by \cref{thm:lower}.

Unfortunately, these higher-order terms in the Taylor series cannot be ignored
at larger values of $t$. As we have $\norm{H}_\infty \le 1$, we can bound the
$k\nth$ order term of the Taylor series expansion of $H$ by $\bO{t^k}$, and so
we must set $t$ to be at most $\bT{\eps}$, resulting in the total time
evolution of $\bO{\eps^3}$ obtained in previous work \cite{gutierrez2024simplealgorithmstestlearn,arunachalam2024testinglearningstructuredquantum}.

To evade this barrier, we will instead show that it is possible to
(approximately) simulate evolving by $H_{> k}$, which is composed of only the
\emph{non-local} terms of the Pauli decomposition of $H$. Note that if $H$ is
$k$-local, this is $0$, while if it is not, $H_>k$ is the difference between
$H$ and the closest $k$-local Hamiltonian. Suppose we could evolve by the time
evolution operator of this Hamiltonian.  Then performing the Bell sampling
procedure from before would return $\ket{\maxentangled}$ with probability
\begin{align*}
	&\left|\bra{\maxentangled} \left(I^{\otimes n}\otimes e^{-i H_{> k} t} \right)\ket{\maxentangled}\right|^2\\
	&= \left|\bra{\maxentangled} \left(I^{\otimes n}\otimes \left(\sum_{\ell=0}^\infty \left(H_{> k}\right)^\ell\frac{(it)^\ell}{\ell!} \right)\right) \ket{\maxentangled}\right|^2\\
	&= \left|1 + \bra{\maxentangled} \left(I^{\otimes n}\otimes \left(\sum_{\ell=2}^\infty \left(H_{> k}\right)^\ell\frac{(it)^\ell}{\ell!} \right)\right) \ket{\maxentangled}\right|^2\\
&= 1 - \bra{\maxentangled}\left(I^{\otimes n} \otimes \left(H_{>k}\right)^2\right)\ket{\maxentangled} + \sum_{\ell=3}^\infty\bO{t^\ell \cdot \abs*{\bra{\maxentangled}
\left(I^{\otimes n}\otimes \left(H_{> k}\right)^\ell\right) \ket{\maxentangled}}}
\end{align*}
as $H$ contains no identity term.

To tame this infinite series, imagine that $\norm{H_{> k}}_\infty \leq 1$ (we
will eventually evolve by a related operator $A$ that \emph{does} satisfy
$\norm{A}_\infty \leq 1$). Then we have \[
\abs*{\bra{\maxentangled}
\left(I^{\otimes n}\otimes \left(H_{> k}\right)^\ell \right)
\ket{\maxentangled}} \leq \bra{\maxentangled} \left(I^{\otimes n}\otimes
\left(H_{> k}\right)^2 \right) \ket{\maxentangled}\]
for all integers $\ell \geq 2$, so as long as $t$ is a sufficiently small
constant, we have \[
\left|\bra{\maxentangled} \left(I^{\otimes n}\otimes e^{-i H_{> k} t}
\right)\ket{\maxentangled}\right|^2 \ge 1 - 0.99 \cdot \bra{\maxentangled}\left(I^{\otimes n} \otimes \left(H_{>k}\right)^2\right)\ket{\maxentangled} = 1 - 0.99 \cdot \tra\left((H_{> k})^2\right)/2^n
\]
where $\tra\left((H_{> k})^2\right)/2^n = \eps^2$ is exactly the squared
normalized Frobenius distance of $H$ from being $k$-local.  So if we apply
$e^{-i H_{> k} t}$ with $t = \bT{1}$, we are left with a $\approx \eps^2$
probability of sampling a non-local Pauli term if $H$ is non-local, and are
guaranteed to measure identity if $H$ is local (as then $e^{-iH_{>k} \cdot t}$
is the identity).  This means we can distinguish locality from non-locality
with $\bO{\eps^{-2}}$ repetitions, requiring $\bO{\eps^{-2}}$ total evolution
time.\footnote{Unfortunately, even with access to the time evolution operator
of $H_{>k}$ we cannot set $t$ to the optimal $\bT{1/\eps}$, as we lose control
of the higher-order terms of the Taylor expansion.}
	
Now, we cannot actually apply $e^{-i H_{> k} t}$.  However, when starting at
$\ket{\maxentangled}$, we can approximate it up to $t = \bT{1}$ by the
use of a process reminiscent of the Elitzur-Vaidman bomb-tester
\cite{elitzur1993bomb} and Quantum Zeno effect \cite{Facchi_2008}, which we refer to as \emph{Trotterized postselection}.

Let $D$ be the subspace of Bell states corresponding to non-local Paulis
\emph{or} identity and let $\Pi_D$ be the projector onto that subspace.
Starting with $\ket{\maxentangled}$ once again, we apply $I^{\otimes n} \otimes
e^{-i H t'}$ for $t' = \bO{\eps}$, measure with $\{\Pi_D, I^{\otimes 2n} - \Pi_D\}$,
and then post-select on the measurement result $\Pi_D$.  We then repeat our
application of $I^{\otimes n} \otimes e^{-i H t'}$ and postselection, for
$\bO{1/t'}$ iterations, provided our postselection succeeds each time.

As we start with $\ket{\maxentangled}$, then make small adjustments (i.e.,
$e^{-i H t} \approx I^{\otimes 2n}$ for small $t$), the chance of failing the
postselection is small: only $\bO{\eps^2}$ at each iteration, and so as long
as we only use $\bO{1/\eps}$ iterations, we will succeed with probability $1 -
\bO{\eps}$.  Now, as we are taking small steps, we can approximate each
iteration of $\Pi_D \left(I^{\otimes n} \otimes e^{-i H \cdot \bO{\eps}}
\right)  \Pi_D$ as \[
\Pi_D \left(I^{\otimes n} \otimes e^{-i H \cdot \bO{\eps}} \right)  \Pi_D =
\Pi_D \left(I^{\otimes n} \otimes \sum_{\ell=0}^\infty H^\ell \frac{(-i)^\ell
\bO{\eps^\ell}}{\ell!} \right)  \Pi_D = e^{-i A \cdot \bO{\eps}} + R
\] where we define $A \coloneqq \Pi_D (I^{\otimes n} \otimes H) \Pi_D$
and choose some $\norm{R}_\infty \leq \bO{\eps^2}$.\footnote{Note that the $\Pi_D$ on the right does nothing
besides make $A$ obviously Hermitian, assuming our invariant of our
postselection succeeding.} 

Now, in general, $A \neq I^{\otimes n} \otimes  H_{> k}$, but as long as $H$
has no identity term in its Pauli decomposition\footnote{We can assume this
without loss of generality, as our algorithm never uses controlled application
of $e^{-iH\cdot t}$, and so any identity term would manifest as an undetectable
global phase.}, by
construction $A \ket{\maxentangled} = \left(I^{\otimes n} \otimes  H_{>
k}\right)\ket{\maxentangled}$, and so $\bra{\maxentangled}
A^2\ket{\maxentangled} = \bra{\maxentangled} I \otimes \left(H_{>k}\right)^2
\ket{\maxentangled}$.  Combined with the fact that $\norm{A}_\infty = \norm{\Pi_D \left(I^{\otimes n} \otimes H\right) \Pi_D}_\infty \leq
\norm{H}_\infty \leq 1$, we can argue that, if we iterate $t/t'$ times
\begin{align*}
	\bra{\maxentangled}\prod_{i = 1}^{t/t'} e^{-i A \cdot t'}
\ket\maxentangled &= \bra{\maxentangled}e^{-i A \cdot t} \ket{\maxentangled}\\
&=\bra{\maxentangled} 
\left(\sum_{\ell=0}^\infty A^\ell\frac{(-it)^\ell}{\ell!} \right)
\ket{\maxentangled}\\
 &= 1 - t^2\bra{\maxentangled}H_{>k}^2\ket{\maxentangled} +
\bO{t^3\cdot\eps^2} 
\end{align*}
where the final inequality follows from the fact that for all $k > 2$,
\[\abs*{\bra{\maxentangled} A^k \ket{\maxentangled}} \leq \norm{A}^{k-2}_\infty\bra{\maxentangled} A^2
\ket{\maxentangled} \leq  \bra{\maxentangled} \left(I^{\otimes n}\otimes (H_{>
k})^2 \right)\ket{\maxentangled} = \eps^2.\]
So as our method based on access to the time evolution operator of $H_{>k}$
only required distinguishing between $\bra{\maxentangled}
H_{>k}\ket{\maxentangled}$ being $\bT{\eps^2}$ and $0$ we can emulate it with
access to $e^{-iAt}$ without losing too much accuracy, as long as we take $t$
to be a small enough constant. We can therefore test locality with a total time
evolution of $\bO{\eps^{-2}}$.

\subsection{Lower Bound}
To prove the lower bound, it suffices to show that for any $k$ there exists
Hamiltonians $H_1$ and $H_2$ such that a query to the time $t$ evolution of
$H_1$ and $H_2$ differ in diamond distance by at most $\bO{(\eps_2 -
\eps_1)t}$, with $H_1$ $\eps_1$-close to being $k$-local and
$H_2$ $\eps_2$-far from being $k$-local. 

We achieve this by considering the weight-$k$ Pauli $Z_{1:k}$ that is $Z$ on
the first $k$ qubits, and identity on the last $n-k$ qubits. We then set $H_1
\coloneqq \eps_1 Z_{1:k}$ and $H_2 \coloneqq \eps_2 Z_{1:k}$.  Because
$Z_{1:k}$ is diagonal, so is $e^{-i\eps Z_{1:k} \cdot t}$, making it straightforward
to bound the diamond distance of the two time evolution operators by
$\bO{t(\eps_2 - \eps_1)}$.  By the sub-additivity of diamond distance, the
total time evolution required to distinguish the two Hamiltonians with constant
probability is therefore at least $\bOm{(\eps_2-\eps_1)^{-1}}$.

\section{Preliminaries}\label{sec:prelim}
\subsection{Quantum Information}
A Hamiltonian on $n$-qubits is a $2^n \times 2^n$ Hermitian matrix. The time
evolution operator of a Hamiltonian $H$ for time $t \geq 0$ is the unitary
matrix
\[
	e^{-i H t} \coloneqq \sum_{k=0}^\infty H^k (-i)^k \frac{t^k}{k!}.
\]
We define the $n$-qubit Pauli matrices to be $\pauli \coloneqq \{I, X, Y,
Z\}^{\otimes n}$, where \[
X = \paren*{\begin{tabular}{cc} $0$ & $1$\\ $1$ & $0$\end{tabular}}, Y =
\paren*{\begin{tabular}{cc} $0$ & $-i$\\ $i$ & $0$\end{tabular}}, Z =
\paren*{\begin{tabular}{cc} $1$ & $0$\\ $0$ & $-1$\end{tabular}}.
\]
For any Pauli $P$, we denote the locality $\abs{P}$ to be
the number of non-identity terms in the tensor product.  Let the Frobenius
inner product between matrices $A$ and $B$ be $\langle A , B \rangle \coloneqq
\tra(A^\dagger B)$.  The orthogonality of Pauli matrices under the Frobenius
inner product is implied by the fact that any product of Paulis is another
Pauli (up to sign) and the fact that among them only the identity has non-zero
trace.  Given a matrix $A = \sum_{P \in \pauli} \alpha_P P$, the locality of
$A$ is the largest $\abs{P}$ such that $\alpha_P \neq 0$.  If $A$ is a
Hamiltonian (i.e., Hermitian) then all $\alpha_P$ are real-valued.  The
\emph{normalized} Frobenius norm is given by \[
\lVert A \rVert_2 = \sqrt{\frac{\langle A, A \rangle}{2^n}} =
\sqrt{\frac{\tra(A^\dagger A)}{2^{n}}} = \sqrt{\sum_{P \in \pauli}
\abs{\alpha_P}^2},
\]
and will be used as our distance to $k$-locality, in keeping with the previous
literature \cite{bluhm2024hamiltonianpropertytesting,
gutierrez2024simplealgorithmstestlearn,
arunachalam2024testinglearningstructuredquantum}.  The other important norm will
be the (unnormalized) spectral norm $\lVert A \rVert_\infty$, which is the
largest singular value of $A$.  For any matrix $A$, $\norm{A}_2 \leq
\norm{A}_\infty$, recalling that $\norm{\cdot}_2$ is the \emph{normalized}
Frobenius norm.  As a form of normalization and to be consistent with the
literature, we will assume that $\lVert H \rVert_\infty \leq 1$ for any
Hamiltonian referenced.  We will also WLOG assume that $\tra(H) = 0$ for any
Hamiltonian, since it does not affect the time evolution unitary beyond a
global phase, and so as our algorithms do not use controlled application of the
unitary, they cannot be affected by it.

We define $A_{> k} \coloneqq \sum_{\abs{P} > k} \alpha_P P$ and subsequently
$A_{\leq k} \coloneqq \sum_{\abs{P} \leq k} \alpha_P P$.  By the orthogonality
of the Pauli matrices under the Frobenius inner product, $A_{\leq k}$ is the
$k$-local Hamiltonian that is closest to $A$ with distance $\lVert A - A_{\leq k}\rVert_2 =
\lVert A_{> k} \rVert_2$.

Let $B = \set{\ket{\Phi^+}, \ket{\Phi^-}, \ket{\Psi^+}, \ket{\Psi^-}}$ denote
the set containing the four Bell states.  We will view $B^{\otimes n}$ as a
basis of $\Cbb^{2^n} \otimes \Cbb^{2^n}$, in which for each copy of $B$, one
qubit is assigned to the left register and one to the right.  Note that, up to
phase, every state in $B^{\otimes n}$ is equal to $\paren{I^{\otimes n} \otimes
P}\ket{\Phi^+}^{\otimes n}$ for a unique $P \in \pauli$.  We will write
$\ket{\sigma_P}$ for this basis element.  As an example,
\[\ket{\Phi^+}^{\otimes n} = \ket{\sigma_{I^{\otimes n}}} =
\frac{1}{\sqrt{2^n}}\sum_{x \in \{0, 1\}^n} \ket{x}\otimes \ket{x}.
\]
If $U = \sum_{P \in \pauli} \alpha_P P$ is a unitary matrix, then by Parseval's
identity, $\sum_{P \in \pauli} \abs{\alpha_P}^2 = 1$, i.e.\ $\abs{\alpha_P}^2$
gives a probability distribution over the Paulis.  Applying $I^{\otimes n}
\otimes U$ to the state $\ket{\sigma_{I^{\otimes n}}} = \ket{\Phi^+}^{\otimes n}$ and measuring in the Bell
basis  $B^{\otimes n}$ allows one to sample from this distribution
\cite{montanaro2010quantumbooleanfunctions}.

For a quantum channel that takes as input an $n$-qubit state, we will let the
diamond norm refer to $\lVert \Lambda \rVert_\diamond \coloneqq \max_\rho
\lVert (I^{\otimes n} \otimes \Lambda)(\rho) \lVert_1$ where the maximization
is over all $2n$-qubit states $\rho$.  The diamond distance famously
characterizes the maximum statistical distinguishability (i.e., induced trace
distance) between quantum channels \cite[Section 9.1.6]{wilde2017quantum}, even
with ancillas.

\subsection{Probability}
\begin{fact}[Multiplicative Chernoff Bound]\label{fact:chernoff}
	Suppose $X_1, \dots, X_m$ are independent Bernoulli random variables. Let $X$ denote their sum and let $\mu \coloneqq \expect{X}$. Then for any $t \ge 0$
	\[
	\Pr\left[X \leq (1- t)\mu\right] \leq e^{-t^2 \mu / 2}.
	\]
\end{fact}
We will not need a particularly tight form of this bound, so for ease of
analysis we state the following (loose) corollary.
\begin{corollary}\label{corollary:bernoulli_bound_mult}
	Suppose $X_1, \dots, X_m$ are i.i.d.\ Bernoulli random variables with probability $p$, and \[m = \frac{2}{p}\left(d + \log(1/\delta)\right).\] Then
	\[
	\Pr\left[\sum_{i=1}^m X_i < d \right] \leq \delta.
	\]
\end{corollary}
\begin{proof}
	Let $\mu \coloneqq \expect{\sum_{i=1}^m X_i} = mp$ and let $\gamma \coloneqq 1 - \frac{d}{\mu}$.
	By the \nameref{fact:chernoff}, 
	\begin{align*}
		\Pr\left[\sum_{i = 1}^m X_i < d\right]
		&= \Pr\left[\sum_{i = 1}^m X_i < (1- \gamma) \mu\right]\\
		&\leq \exp\left( -\frac{\mu}{2}\gamma^2 \right) \\
		&= \exp\left( -\frac{\mu}{2} - \frac{d^2}{2\mu} + d  \right) \\
		&\leq \exp\left( -\frac{mp}{2} + d  \right).
	\end{align*}
	Hence, as long as 
	\[
	m \geq \frac{2\log(1/\delta) + 2d}{p},
	\]
	then $\sum_{i=1}^m X_i \leq d$ with probability at most $\delta$.
\end{proof}

\begin{fact}[Bernstein's inequality]\label{fact:bernstein}
	Suppose $X_1, \dots, X_n$ are independent Bernoulli random variables. Let $X$ denote their sum and let $\mu$ and $\sigma^2$ be the expectation and variance of $X$ respectively. Then for any $t \ge 0$,
	\[
	\Pr\left[X - \mu \geq t \right] \leq e^{-\frac{\frac{t^2}{2}}{ \sigma^2 + \frac{t}{3}}} \text{ and } 	\Pr\left[X - \mu \leq -t\right] \leq  e^{-\frac{\frac{t^2}{2}}{ \sigma^2 + \frac{t}{3}}} .
	\]
\end{fact}

\section{Upper Bound}
We will frequently use the truncation of the Taylor series of the matrix exponential to analyze our algorithm.
The following will allow us to then bound the error of the truncation.

\begin{fact}[{\cite[Lemma F.2]{childs2018simulation}}]\label{fact:taylor-trunc}
	If $\lambda \in \mathbb{C}$ then $\left| \sum_{k=\ell}^\infty \frac{\lambda^k}{k!} \right| \leq \frac{\abs{\lambda}^{\ell}}{\ell!}e^{\abs{\lambda}}$.
\end{fact}

\begin{corollary}\label{cor:taylor-truncation}
	For $n$-qubit Hamiltonian $H$ with $\lVert H \rVert_\infty \leq 1$, the first order Taylor series expansion of the matrix exponential gives \[e^{-i H t} = I^{\otimes n}  - i Ht + \frac{e^t \cdot t^2}{2}R\]
	for $\lVert R \rVert_\infty \leq 1$.
\end{corollary}
\begin{proof}
	By the triangle inequality and the fact that $\norm{H^k}_\infty \leq
	\norm{H}_\infty \leq 1$ for $k \geq 1$:
	\[\norm{e^{-i H t} - (I^{\otimes n} - iHt)}_\infty = \norm*{\sum_{k=2}^\infty (-i)^k\frac{H^k t^k}{k!}}_\infty \leq \sum_{k=2}^\infty \frac{\norm{H^k}_\infty t^k}{k!} \leq \sum_{k=2}^\infty \frac{t^k}{k!} \leq \frac{e^t \cdot t^2}{2},\]
	using \cref{fact:taylor-trunc} at the end.
	Setting $R \coloneqq \frac{2}{e^t \cdot t^2} \left(e^{-i H t} - (I^{\otimes n} - iHt)\right)$ completes the proof.
\end{proof}

We also prove the related fact to bound the real and imaginary terms.

\begin{fact}\label{fact:taylor-trunc-real-and-imaginary}
	If $\lambda \in \mathbb{C}$ then $\left| \sum_{k=\ell}^\infty \frac{\lambda^{2k}}{(2k)!} \right| \leq \frac{\abs{\lambda}^{2\ell}}{(2\ell)!} \cosh(\abs{\lambda})$ and $\left| \sum_{k=\ell}^\infty \frac{\lambda^{2k+1}}{(2k+1)!} \right| \leq \frac{\abs{\lambda}^{2\ell+1}}{(2\ell+1)!} \cosh(\abs{\lambda})$.
\end{fact}
\begin{proof}
	\begin{align*}
		\left| \sum_{k=\ell}^\infty \frac{\lambda^{2k}}{(2k)!} \right|
		 \leq  \sum_{k=\ell}^\infty \frac{\abs{\lambda^{2k}}}{(2k)!} 
		= \abs{\lambda}^{2\ell} \sum_{k=0}^\infty \frac{\abs{\lambda}^{2k}}{(2k + 2\ell)!}
		\leq \frac{\abs{\lambda}^{2\ell}}{(2\ell)!} \sum_{k=0}^\infty \frac{\abs{\lambda}^{2k}}{(2k)!}
		= \frac{\abs{\lambda}^{2\ell}}{(2\ell)!} \cosh(\abs{\lambda})
	\end{align*}
	and
	\begin{align*}
	\left| \sum_{k=\ell}^\infty \frac{\lambda^{2k+1}}{(2k+1)!} \right|
	&\leq  \sum_{k=\ell}^\infty \frac{\abs{\lambda^{2k+1}}}{(2k+1)!} 
	= \abs{\lambda}^{2 \ell+1} \sum_{k=0}^\infty \frac{\abs{\lambda}^{2k}}{(2k + 2\ell + 1)!}\\
	&\leq \frac{\abs{\lambda}^{2\ell+1}}{(2\ell+1)!} \sum_{k=0}^\infty \frac{\abs{\lambda}^{2k}}{(2k)!}
	= \frac{\abs{\lambda}^{2\ell+1}}{(2\ell+1)!} \cosh(\abs{\lambda}). \qedhere
	\end{align*}
\end{proof}

\subsection{Algorithm}

\begin{definition}
	We will use $D$ to denote the subspace of $\Cbb^{2^n} \otimes \Cbb^{2^n}$
	spanned by $\ket{\sigma_P}$ for Pauli strings $P$ that are either the identity or
	are \emph{not} $k$-local, and $\Pi_D$ to denote the projector onto $D$.
	We define $A \coloneqq \Pi_D \left(I^{\otimes n} \otimes H\right) \Pi_D$.
\end{definition}

We start by giving an algorithm
that returns a Bernoulli random variable $X \in
\{0, 1\}$, where $\expect{X}$ approximates the distance of $H$ from being
$k$-local.
It does so by iteratively applying $e^{-i \alpha H}$ sandwiched by $\{\Pi_D, I^{\otimes 2n} - \Pi_D\}$ measurements.

\begin{algorithm}
	\caption{Hamiltonian Locality Estimator via Trotterized Postselection}
	\begin{algorithmic}[1]
		\State Start with $\ket{\phi} = \ket{\maxentangled}$.
		\For{$\frac{50}{\sqrt{\eps_2^2-\eps_1^2}}$ iterations}
		\label{ln:iteration}
    \State Apply $(I^{\otimes n} \otimes e^{-i \alpha H}$ to $\ket{\phi}$ for
    $\alpha = \frac{\eps_2^2 - \eps_1^2}{100 \eps_2}$. \label{ln:evolve}
		\State Measure $\ket{\phi}$ with the projectors $\Pi_D, I^{\otimes 2n} - \Pi_D$, terminating
		and returning $\perp$ if the result is $I^{\otimes 2n} - \Pi_D$. \label{ln:truncate} 
		\EndFor
		\State Measure $\ket{\phi}$ in the Bell basis, returning $0$ if the result is
		$\ket{\maxentangled}$ and $1$ otherwise.
	\end{algorithmic}
	\label{alg:primitive}
\end{algorithm}

Let $\alpha \coloneqq \frac{\eps_2^2 - \eps_1^2}{100 \eps_2}$ be the step-size
used in \cref{ln:evolve}, $t \coloneqq
\frac{\sqrt{\eps_2^2-\eps_1^2}}{2\eps_2}$ be the total time evolution used in
\cref{alg:primitive}, and let $m \coloneqq t/\alpha =
\frac{50}{\sqrt{\eps_2^2-\eps_1^2}}$ be the number of iterations used in
\cref{ln:iteration}.  In our analysis will frequently use the fact that $\alpha
\leq \frac{\eps_2}{100} \leq \frac{1}{100}$ and $t \leq 0.5$ to simplify
higher-order terms.

\begin{remark}
	While we attempted to keep the constants in the algorithm reasonable, no attempt was made to optimize them.
	We observe that $t$ should remain $\bT{\frac{\sqrt{\eps_2^2-\eps_1^2}}{\eps_2}}$ for optimal scaling, but $\alpha$ can be made arbitrarily small to (marginally) improve the constants in the total time evolution used.
	This has a cost in the total number of queries used, scaling roughly proportional to $\alpha^{-1}$.
\end{remark}

First we show that the final state of the Trotterized postselection algorithm corresponds to evolving
$\ket{\maxentangled}$ by $e^{-i A t}$, with a bounded error term.
There are two main sources of error: (1) the error from higher-order terms in the respective Taylor series of $e^{-i A \alpha}$ and $\Pi_D \left(I^{\otimes n} \otimes e^{-i H \alpha}\right) \Pi_D$ not matching and (2) the error from postselection causing normalization issues.
The following technical lemma allows us to tackle the error from (1).
This is done by showing that $e^{-i t A} = \Pi_D \left(I^{\otimes n} \otimes e^{-i t H}\right) \Pi_D \pm \bO{\alpha^2}$ for sufficiently small $\alpha$.
By chaining these together, the triangle inequality will eventually show in \cref{lm:final_state} that the accumulated error is then at most $\bO{\alpha^2 m} = \bO{\alpha t}$.

\begin{lemma}
	\label{lm:truncate_ev}
	Let $H = \sum_{P \in \pauli} \alpha_P P$ be any Hamiltonian with
	$\norm{H}_\infty \le  1$.
	Then,
	\[
	\Pi_D (I^{\otimes n} \otimes e^{-i\alpha H}) \Pi_D = e^{-i \alpha A} + \eta
	\] where $\norm{\eta}_\infty \leq e^\alpha \cdot \alpha^2$.
\end{lemma} 
\begin{proof}
	By Taylor expanding the complex exponential of $e^{-i \alpha H}$ and applying \cref{cor:taylor-truncation}, we get 
	\begin{align*}
		\Pi_D (I^{\otimes n} \otimes e^{-i \alpha H}) \Pi_D
		&= \Pi_D\left(I^{\otimes n} \otimes \paren*{I^{\otimes n} - i \alpha H + \frac{e^\alpha \cdot \alpha^2}{2}R}\right) \Pi_D\\
		&= I^{\otimes 2n} - i\alpha A + \frac{e^t \cdot \alpha^2}{2}R^\prime
	\end{align*}
	where $\lVert R^\prime \rVert_\infty \le \lVert I^{\otimes n} \otimes R
	\rVert_\infty = \lVert R \rVert_\infty \le 1$.
	
	Next, we observe that $\lVert A \rVert_\infty  \leq \lVert I^{\otimes n}
	\otimes H \rVert_\infty = \lVert H \rVert_\infty \leq 1$ and that $A$ is
	Hermitian by symmetry.  We can then Taylor expand $e^{-i\alpha A}$ to get
	\[
	e^{-i \alpha A} = I^{\otimes 2n} - i\alpha A +\frac{e^\alpha \cdot \alpha^2}{2} Q
	\]
	where $\lVert Q \rVert_\infty \leq 1$.  By the triangle inequality, the
	difference
	\[\eta \coloneqq \Pi_D (I^{\otimes n} \otimes e^{-i\alpha H})\Pi_D -
	e^{-i \alpha A}
	\] 
	between these two linear transformations satisfies
	\[
	\lVert \eta \rVert_\infty \leq \lVert R^\prime\rVert_\infty \cdot \frac{e^\alpha
		\cdot \alpha^2}{2} + \lVert Q \rVert_\infty \cdot \frac{e^\alpha \cdot \alpha^2}{2} \le e^\alpha
	\cdot \alpha^2. \qedhere
	\]
\end{proof}

Luckily, the error from (2) is mostly a non-issue, using a process similar to
the Elitzur-Vaidman bomb \cite{elitzur1993bomb}: by taking small steps between
applications of $\Pi_D$, we ensure that we are barely changing our system, and
so the postselection nearly always succeeds.  This also means that the
normalization error can be suppressed to be arbitrarily small, at the cost of
linearly increasing the number of times we have to query the time evolution
operator.  Using these facts together, we show that \cref{alg:primitive}
approximately applies the time evolution operator of $A$.

\begin{lemma}
	\label{lm:final_state}
	\sloppy
	\cref{alg:primitive} terminates before the final measurement with probability at most
	$\frac{99}{98}\alpha t$. If it does not, $\ket{\phi} =  e^{-i A t}\ket{\maxentangled}
	+ \ket{\Delta}$ just before the final measurement, with $\norm{\ket{\Delta}}_2
	\leq \frac{7}{4}\alpha t $.
\end{lemma}
\begin{proof}	
	Note that the algorithm can only be terminated early if, in one of the loop
	iterations, the measurement in \cref{ln:truncate} returns $I^{\otimes 2n} - \Pi_D$. At the
	start of the iteration $\ket{\phi} = \ket{\maxentangled} \in D$.  Since
	$\ket{\phi}$ remains within $D$ after each successful iteration, by Taylor
	expanding the exponential, and applying \cref{cor:taylor-truncation}
	to obtain a suitable $R$ with $\norm{R}_\infty \le 1$, the probability of
	failure at each iteration is at most 
	\begin{align*}
		&\norm*{(I^{\otimes 2n} - \Pi_D) \left(I^{\otimes n} \otimes e^{-iH\alpha }\right) \Pi_D \ket{\phi}}_2^2\\
		&= \norm*{(I^{\otimes 2n} - \Pi_D) \left(I^{\otimes n} \otimes \left(I^{\otimes n} -
			i\alpha H + \frac{\alpha^2}{2}
			e^{\alpha}R\right)\right)\ket{\phi}}_2^2\\
		&= \norm*{(I^{\otimes 2n} - \Pi_D) \left(- i \alpha (I^{\otimes n}
			\otimes H) + \frac{\alpha^2}{2}e^\alpha (I^{\otimes n} \otimes
			R)\right)\ket{\phi}}_2^2\\
		& \leq \left(\alpha \lVert H \rVert_\infty + \frac{\alpha^2 e^{\alpha}}{2} \lVert R \rVert_\infty\right)^2\\
		&\leq  \left(1 + \alpha e^\alpha + \frac{\alpha^2}{4}e^{2\alpha}\right)\alpha^2 \\
		&< \frac{99}{98}\alpha^2
	\end{align*}
  where the third line follows from $\ket \phi \in D$, the fourth from the
  triangle inequality combined with the definition of the spectral norm, and
  the final line from $\alpha \le 0.01$.  By a union bound over the $m$
  iterations, the first part of the lemma follows, noting that $t \coloneqq
  \alpha \cdot m$.
	
  For the second part pertaining to accuracy, first we note that in each
  iteration, if the measurement in \cref{ln:truncate} does \emph{not} make the
  algorithm terminate, the iteration had the effect of taking $\ket{\phi} \in
  D$ to \[\Pi_D\left(I^{\otimes n} \otimes e^{-i \alpha H}\right)\ket{\phi} =
  \Pi_D\left(I^{\otimes n} \otimes e^{-i \alpha H}\right) \Pi_D \ket{\phi},\]
  normalized to length $1$.  After the $m$ iterations of the loop of
  \cref{alg:primitive}, $\ket{\phi}$ is then
	\[
		\prod_{i=1}^{m} \Pi_D \left(I^{\otimes n} \otimes e^{-i\alpha H}\right) \Pi_D \ket{\maxentangled}
	\]
	normalized to length $1$.
	By \cref{lm:truncate_ev}, before normalization this is equivalent to
	\[
		\prod_{i=1}^{m} \left(e^{-i \alpha A} + \eta\right) \ket{\maxentangled} =  \left( \sum_{k=0}^{m} \binom{m}{k}  e^{-i \alpha A (m-k)}\cdot \eta^k\right) \ket{\maxentangled}
	\]
	for $\norm{\eta}_\infty \leq \alpha^2 e^{\alpha}$.
	The distance of the un-normalized vector from $e^{-i A t}\ket{\maxentangled}$ is then
	\begin{align*}
		\norm*{ e^{-i A t}\ket{\maxentangled} - \prod_{i=1}^{m} \left(e^{-i A t} + \eta\right)\ket{\maxentangled}}_2
		&= \norm*{ \left( \sum_{k=1}^{m} \binom{m}{k}  e^{-i \alpha A (m-k)}\cdot \eta^k\right) \ket{\maxentangled} }_2\\
		&\leq \sum_{k=1}^m m^k  \norm{\eta}_\infty^k\\
		&\leq \sum_{k=1}^m \left(m \alpha^2 e^{\alpha}\right)^k\\
		&\leq \sum_{k=1}^\infty \left(m \alpha^2 e^{\alpha}\right)^k\\
		&= m \alpha^2 e^{\alpha} \frac{1}{1-m \alpha^2 e^{\alpha}}\\
		&= \alpha t e^{\alpha} \frac{1}{1- \alpha t e^{\alpha}}.
	\end{align*}
  Finally, to bound the error introduced by normalization, for each $r \in \brac{m}$, write \[
  \ket{\phi_r} \coloneqq \prod_{i = 1}^{r} \Pi_D (I^{\otimes n} \otimes
  e^{-i\alpha H})\Pi_D \ket{\sigma_{I^{\otimes n}}}
  \]for the projected state at iteration $r$.
  We note that, by the
  same argument proving that the probability of the measurement at any given
  step returning the $I^{\otimes 2n} - \Pi_D$ result is at most $\frac{99}{98}
  \alpha^2$,
  $\ket{\phi_r}$ is separated from $e^{-iAt}\ket{\phi_{r-1}}$ by an
  \emph{orthogonal} vector of length at most $\sqrt{\frac{99}{98}}\alpha
  \norm{e^{-i A t}\ket{\phi_{r-1}}}_2 = \sqrt{\frac{99}{98}}\alpha
  \norm{\ket{\phi_{r-1}}}_2$.  Therefore, \[ 
  \norm{\ket{\phi_r}}_2 \ge \norm{\ket{\phi_{r-1}}}_2 \sqrt{1 - \frac{99}{98}\alpha^2}
  \ge  \norm{\ket{\phi_{r-1}}}_2 - 0.6 \frac{99}{98}\alpha^2
\]
  where the last inequality follows from the fact that $1-\sqrt{1-x} \leq 0.6
  x$ for $x \in [0, \frac{5}{9}]$ and $\frac{99}{98}\alpha^2 < \frac{5}{9}$.
  The total additional error from the normalization is then at most
  $\frac{297}{490} \alpha^2 m = \frac{297}{490}\alpha t$.  By the triangle inequality, the total distance
  from $e^{-i A t}\ket{\maxentangled}$ is at most
	\[
	\frac{297}{490}\alpha t + t \alpha e^{\alpha} \frac{1}{1-\alpha t e^{\alpha}}   \leq \frac{7}{4}\alpha t. \qedhere\]

\end{proof}

We now show that (approximately) applying $e^{-i A t}$ instead of $I^{\otimes n} \otimes
e^{-iHt}$ allows us to suppress the higher-order terms that were preventing us
from increasing the evolution time $t$ when testing for locality. We will need
the following results that let us characterize the individual terms of the
Taylor expansion.

\begin{fact}\label{fact:bell-trace-trick}
	For any matrix $M$, $\bra{\sigma_P}(I \otimes M)\ket{\sigma_Q} = \frac{\tra(PMQ)}{2^n}$.
\end{fact}
\begin{proof}
	\begin{align*}
		\bra{\sigma_P}(I \otimes M)\ket{\sigma_Q}
		&= \frac{1}{2^n}\sum_{x, y \in \{0, 1\}^n} (\bra{x} \otimes \bra{x} P) \left(\ket{y} \otimes M Q \ket{y}\right)\\
		&= \frac{1}{2^n}\sum_{x, y \in \{0, 1\}^n} \langle x | y \rangle \cdot \bra{x} PMQ \ket{y}\\
		&=\frac{1}{2^n}\sum_{x \in \{0, 1\}^n} \bra{x} PMQ \ket{x}\\
		&= \frac{\tra(PMQ)}{2^n} && \qedhere
	\end{align*}
\end{proof}

\begin{lemma}\label{lem:trace-A-zero}
	$\bra{\maxentangled} A \ket{\maxentangled} = 0$.
\end{lemma}
\begin{proof}
	\begin{align*}
		\bra{\maxentangled} A \ket{\maxentangled} &= \bra{\maxentangled} \Pi_D \left(I^{\otimes n} \otimes H\right) \Pi_D \ket{\maxentangled}\\
		&= \bra{\maxentangled} I^{\otimes n} \otimes H \ket{\maxentangled}\\
		&= \frac{1}{2^n}\tra\left(H\right) && (\text{\cref{fact:bell-trace-trick}})\\
		&= 0,
	\end{align*}
	recalling that we have assumed that $\tra(H) = 0$.
\end{proof}

\begin{lemma}\label{lem:trace-A-squared}
	For $k \geq 2$, $\lvert \bra{\maxentangled} A^k \ket{\maxentangled} \rvert \leq \bra{\maxentangled} A^2 \ket{\maxentangled} = \norm{H_{> k}}_2^2$.
\end{lemma}
\begin{proof}
	The first inequality follows because $\lVert A \rVert_\infty \leq
	\lVert H \rVert_\infty \leq 1$, and the fact that $H$ is Hermitian and so
	$A$ is too, meaning that every eigenvalue of $A^k$ is non-increasing in
	magnitude as a function of $k$, and non-negative when $k$ is even.
	
	For the second equality, we observe that \[A \ket{\maxentangled} = \Pi_D \left(I^{\otimes n} \otimes H\right) \Pi_D \ket{\maxentangled} = \Pi_D \left(I^{\otimes n} \otimes H\right)\ket{\maxentangled} = \left(I^{\otimes n} \otimes H_{> k}\right)\ket{\maxentangled},\]
	as $H$ has no identity component.
	By \cref{fact:bell-trace-trick}, 
	\[
	\bra{\maxentangled} A^2 \ket{\maxentangled} = \bra{\maxentangled} I^{\otimes n} \otimes (H_{> k})^2 \ket{\maxentangled} = \frac{1}{2^n}\tra\left((H_{> k})^2\right) = \norm{H_{> k}}_2^2. \qedhere
	\]
\end{proof}

Combining \cref{lm:final_state,lem:trace-A-zero,lem:trace-A-squared}, we are able to give bounds on the acceptance probability of \cref{alg:primitive} (assuming it does not terminate early) based on how close or far $H$ is from being $k$-local.
This gives us an algorithm for testing locality, through repetition of \cref{alg:primitive} and concentration of measure.

\begin{lemma}
	\label{lm:expectcorrect}
  Let $\eps \coloneqq \norm{H_{> k}}_2$.  The probability that
  \cref{alg:primitive} outputs $1$, conditioned on not terminating early, is at
  least $\eps^2 t^2 \left(1- \frac{t^2}{10} - \frac{13}{50} \eps^2 t^2\right) - \frac{7}{2}
  \eps \alpha t^2$ and no more than $\eps^2 t^2 \left(1 +
  \frac{1}{10}t^2\right) + \frac{287}{80} \eps \alpha t^2 + \frac{49}{1600}\eps_2
  \alpha t^2$.\footnote{The $\eps_2$ in the $\frac{49}{1600}\eps_2
  	\alpha t^2$ term of the upper bound is intended and \emph{not} a typo.}
\end{lemma}
\begin{proof}
	At the end of \cref{alg:primitive} (assuming it did not terminate early), the final state lies in $D$.
	By \cref{lm:final_state} and the definition of the final measurement, the
	probability that the algorithm outputs $1$ is the squared length of the component of \[
		\ket{\psi} \coloneqq e^{-i A t}\ket{\maxentangled} + \ket{\Delta}
	\]
  along the complement of $\ket \maxentangled$, for some $\Delta$ such that
  $\norm{\ket{\Delta}}_2 \leq  2 \alpha t $.  So by the triangle
  inequality
	\[
    \left( \sqrt{1-\abs{\bra{\maxentangled}e^{-i A t} \ket{\maxentangled}}^2} -
    \norm{\ket{\Delta}}_2\right)^2 \leq \Pr\left[X=1\right] \leq \left(
    \sqrt{1-\abs{\bra{\maxentangled}e^{-i A t} \ket{\maxentangled}}^2} +
    \norm{\ket{\Delta}}_2\right)^2.\footnote{One might think to use
    $1-\abs*{\bra{\maxentangled} \left(e^{-i A t} \ket{\maxentangled} +
    \ket{\Delta}\right)}^2$ followed by the triangle inequality, but this
    actually leads to a lossy analysis of the number of queries used.}
	\]
	To analyze $\abs*{\bra{\maxentangled}e^{-i A t}\ket{\maxentangled}}$, we note that because $A$ is Hermitian, $\bra{\maxentangled} A^k \ket{\maxentangled}$ is real-valued for all $k \geq 0$.
	By splitting up the Taylor expansion of the matrix exponential into real and imaginary terms, we see that
	\begin{align*}
		&\abs*{\bra{\maxentangled} e^{-i A t}\ket{\maxentangled}}^2\\
		&= \left| \bra{\maxentangled} \left(\sum_{m=0}^\infty (-i)^m \frac{A^m t^m}{m!}\right) \ket{\maxentangled}\right|^2 \\
		&= \left| \bra{\maxentangled} \left(\sum_{m=0}^\infty (-1)^m \frac{A^{2m}t^{2m}}{(2m)!}\right) \ket{\maxentangled}\right|^2 + \left| \bra{\maxentangled} \left(\sum_{m=0}^\infty (-1)^{m+1} \frac{A^{2m+1}t^{2m+1}}{(2m+1)!}\right) \ket{\maxentangled}\right|^2.
	\end{align*}
Analyzing the first term, we  see that
	\begin{align*}
		&\left| \bra{\maxentangled} \left(\sum_{m=0}^\infty (-1)^m \frac{A^{2m} t^{2m}}{(2m)!}\right) \ket{\maxentangled}\right|\\
    &= \left| \bra{\maxentangled} \left(I^{\otimes 2n} - \frac{t^2}{2}A^2 +
    \sum_{m=2}^\infty (-1)^m \frac{A^{2m} t^{2m}}{(2m)!}\right)
    \ket{\maxentangled}\right| \\
		&=  \left| \frac{\tra(I^{\otimes n})}{2^n} - \frac{t^2}{2} \bra{\maxentangled}A^2\ket{\maxentangled} + \bra{\maxentangled}\left(\sum_{m=2}^\infty(-1)^m \frac{A^{2m} t^{2m}}{(2m)! }\right) \ket{\maxentangled}\right| && (\text{\cref{fact:bell-trace-trick}})\\
		&=  \left|1 - \frac{\eps^2 t^2}{2}  + \sum_{m=2}^\infty(-1)^m \bra{\maxentangled}\frac{A^{2m} t^{2m}}{(2m)!} \ket{\maxentangled}\right| && (\text{\cref{lem:trace-A-squared}})\\
		&= 1 - \frac{\eps^2 t^2}{2} + \eta_{\text{real}} 
	\end{align*}
  where $\abs{\eta_{\text{real}}} \leq \frac{\eps^2 t^4}{24} \cosh(t) \leq
  \frac{\eps^2 t^4}{20}$ by \cref{fact:taylor-trunc-real-and-imaginary},
  \cref{lem:trace-A-squared}, the triangle inequality, and the fact that $t \le
  \frac{1}{2}$.

\noindent
Then, for the second term, we have
	\begin{align*}
    \eta_{\text{imaginary}}  \coloneqq&  \left| \bra{\maxentangled}
    \left(\sum_{m=0}^\infty (-1)^m \frac{A^{2m+1} t^{2m+1}}{(2m+1)!}\right)
    \ket{\maxentangled}\right|\\ =& \left| \bra{\maxentangled} \left(A +
    \sum_{m=1}^\infty (-1)^{m+1} \frac{A^{2m+1} t^{2m+1}}{(2m+1)!}\right)
    \ket{\maxentangled}\right| \\
    =&  \left| \bra{\maxentangled}\left(\sum_{m=1}^\infty(-1)^m \frac{A^{2m+1}
    t^{2m+1}}{(2m+1)! }\right) \ket{\maxentangled}\right| &&
    (\text{\cref{lem:trace-A-zero}})\\ \leq&  \eps^2 \left|\sum_{m=1}^\infty
    \frac{ t^{2m+1}}{(2m+1)!} \right| &&
    (\text{\cref{lem:trace-A-squared},\,triangle inequality})\\
    \leq& \eps^2 \frac{t^3}{6}\cosh(t) \leq \frac{1}{10}\eps^2 t^2 . &&
    (\text{\cref{fact:taylor-trunc-real-and-imaginary}}) 
  \end{align*}
	Since \[
  \abs*{\bra{\maxentangled} e^{-i A t}\ket{\maxentangled}}^2 =
  \left(1-\frac{\eps^2 t^2}{2} + \eta_{\text{real}}\right)^2 +
  \eta_{\text{imaginary}}^2,
  \] we can upper bound it by $\left(1-\frac{\eps^2 t^2}{2} +
  \abs{\eta_{\text{real}}}\right)^2 + \eta_{\text{imaginary}}^2$ and lower
  bound it by $\left(1-\frac{\eps^2 t^2}{2} - \abs{\eta_{\text{real}}}\right)^2$
  as $\eta_{\text{imaginary}} \geq 0$.

	We can therefore upper bound the probability of \cref{alg:primitive} accepting by
	\begin{align*}
		&\left( \sqrt{1-\abs{\bra{\maxentangled}e^{-i A t} \ket{\maxentangled}}^2} + \norm{\ket{\Delta}}_2\right)^2\\
    &\leq \left( \sqrt{1-\left(1-\frac{\eps^2 t^2}{2} -
    \abs{\eta_{\text{real}}}\right)^2} + \frac{7}{4}\alpha t\right)^2 &&
    (\text{\cref{lm:final_state}})\\
	&\leq \paren*{ \sqrt{\eps^2t^2 + 2\abs{\eta_{\text{real}}}} +
	\frac{7}{4}\alpha t}^2\\
    &\leq \eps^2 t^2 + \frac{1}{10}\eps^2 t^4 + \frac{7}{2} \alpha t
    \sqrt{\eps^2 t^2 + \frac{1}{10}\eps^2 t^4} +  \frac{49}{16}\alpha^2 t^2  &&
    \paren*{\abs{\eta_{\text{real}}} \le \frac{\eps^2 t^4}{20}}\\
    &\le \eps^2 t^2 \left(1 +  \frac{1}{10}t^2\right) + \frac{7}{2} \eps \alpha t^2
        \paren*{1 + \frac{1}{10}t^2} +  \frac{49}{16}\alpha^2 t^2\\
		&\leq \eps^2 t^2 \left(1 +  \frac{1}{10}t^2\right) +  \frac{287}{80}\eps \alpha t^2 +  \frac{49}{1600} \eps_2 \alpha t^2   && \left(t \leq 0.5,\, \alpha \leq \frac{\eps_2}{100}\right)
	\end{align*}
	and lower bound it by
	\begin{align*}
		&\left( \sqrt{1-\abs{\bra{\maxentangled}e^{-i A t} \ket{\maxentangled}}^2} - \norm{\ket{\Delta}}_2\right)^2\\
    &\geq \left( \sqrt{1-\left(1-\frac{\eps^2 t^2}{2} +
    \abs{\eta_{\text{real}}}\right)^2 - \eta_{\text{imaginary}}^2} -
    \frac{7}{4}\alpha t\right)^2 && (\text{\cref{lm:final_state}})\\
    &\ge \paren*{\sqrt{\eps^2 t^2 - 2\abs{\eta_{\text{real}}} - \paren*{\frac{\eps^2
    t^2}{2} - \abs{\eta_{\text{real}}}}^2 -  \eta_{\text{imaginary}}^2} -
    \frac{7}{4}\alpha t}^2\\
    &\ge \paren*{\sqrt{\eps^2 t^2 - \frac{\eps^2 t^4}{10}  - \frac{\eps^4
    t^4}{4} - \frac{\eps^4t^4}{100}} -
    \frac{7}{4}\alpha t}^2 && \paren*{\abs{\eta_{\text{real}}} \le \frac{\eps^2
    t^4}{20}, \eta_{\text{imaginary}} \le \frac{\eps^2t^2}{10}}\\
		&\geq \eps^2 t^2 \left(1- \frac{t^2}{10} - \frac{13}{50} \eps^2 t^2\right) - \frac{7}{2} \eps \alpha t^2 &&  
	\end{align*}
  where the last line uses the fact that the the expression inside the square root is at most $\eps^2 t^2$.
\end{proof}

\upperfwdonly*
\begin{proof}
	By \cref{lm:expectcorrect} the output of \cref{alg:primitive}, conditioned on
  succeeding, is a Bernoulli random variable $X_i$ with bounded expectation (we
  will use $i$ to index successful runs of \cref{alg:primitive}).  That is,
  when $\eps \geq \eps_2$ then
	\begin{equation}
    \expect{X_i} \geq \upsilon \coloneqq \eps_2^2 t^2 \left(1- \frac{t^2}{10} -
    \frac{13}{50} \eps_2^2 t^2\right) - \frac{7}{2} \eps_2 \alpha t^2
    \label{eq:upsilon}
	\end{equation}
	and when $\eps \leq \eps_1$ then
	\begin{equation}
    \expect{X_i} \leq \lambda \coloneqq \eps_1^2 t^2 \left(1 +  \frac{1}{10}t^2\right) +
    \frac{287}{80} \eps_1 \alpha t^2 + \frac{49}{1600}\eps_2 \alpha t^2.
    \label{eq:lambda}
	\end{equation}
	Let 
	\begin{align*}
    \tau &\coloneqq \frac{\upsilon + \lambda}{2} = \frac{1}{2}\left[\eps_2^2
    t^2 \left(1- \frac{t^2}{10} -  \frac{13}{50} \eps_2^2 t^2\right) -
    \frac{7}{2} \eps_2 \alpha t^2 + \eps_1^2 t^2 \left(1 +
    \frac{1}{10}t^2\right) + \frac{287}{80} \eps_1 \alpha t^2 +
    \frac{49}{1600}\eps_2 \alpha t^2\right]
	\end{align*}
	then be our decision threshold.  And for convenience let
	\[
		\xi \coloneqq \frac{1}{2}\left[\eps_2^2 t^2 \left(1- \frac{t^2}{10} -  \frac{13}{50} \eps_2^2 t^2\right) - \frac{7}{2} \eps_2 \alpha t^2 - \eps_1^2 t^2 \left(1 +  \frac{1}{10}t^2\right) - \frac{287}{80} \eps_1 \alpha t^2 - \frac{49}{1600}\eps_2 \alpha t^2\right]
	\]
  be $\abs{\tau - \upsilon} = \abs{\tau - \lambda}$.  Observe that, as $\eps_1
  < \eps_2 \leq 1$, $\eps_2 \alpha = \frac{\eps_2^2-\eps_1^2}{100}$ and $t =
  \frac{\sqrt{\eps_2^2-\eps_1^2}}{2\eps_2}$ and so
	\begin{align}
  \xi &\ge  \frac{46}{100}(\eps_2^2-\eps_1^2)t^2 - \frac{23}{100} \eps_2^2 t^4 \nonumber\\
      &\ge \frac{1}{10}\frac{(\eps_2^2-\eps_1^2)^2}{\eps_2^2} \label{eq:xilb}
  \end{align}
  while 
  \begin{equation}
    \xi \le \frac{\eps_2^2 t^2}{2}. \label{eq:xiub}
  \end{equation}
	Now say that we have i.i.d samples $\{X_1, \dots, X_s\}$ from \emph{successful} runs of \cref{alg:primitive} for $s$ to be determined and let $X \coloneqq \sum_{i=1}^s X_i$.
	If $\eps \geq \eps_2$, then by \nameref{fact:bernstein} the probability that $X \leq s \tau$ is at most:
	\begin{align*}
		\Pr\left[\sum_{i=1}^s X_i \leq s\tau\right]
		&= \Pr\left[X - \expect{X} \leq s\tau - \expect{X}\right]\\
		&\leq \exp\left[-\frac{\frac{\left(s\tau - \expect{X}\right)^2}{2}}{s \expect{X}\left(1-\expect{X}\right) + \frac{\expect{X} - s\tau}{3}}\right]\\
		&\leq \exp\left[-\frac{\left(s\tau - \expect{X}\right)^2}{2\left(s \expect{X} + \frac{\expect{X} - s\tau}{3} \right)}\right]\\
		&\leq \exp\left[ - \frac{s^2 \xi^2}{2\left(s\upsilon + s\frac{\xi}{3}\right)}\right]\\
    &\leq \exp\left[ - s\frac{3 \xi^2}{7 \eps_2^2 t^2 }\right] && \paren*{\text{\cref{eq:upsilon}, \cref{eq:xiub}}}\\
    &\leq \exp\left[- \frac{s}{59}
    \frac{(\eps_2^2-\eps_1^2)^3}{\eps_2^4}\right] && \paren*{\text{\cref{eq:xilb}, }t \le \frac{1}{2}}
	\end{align*}
	where the fourth line follows due to the expression in the exponential being monotonically increasing with respect to $\expect{X} \in (\tau, 1]$. Likewise, if $\eps \leq \eps_1$ then the probability that $X \geq s \tau$ is at most:
	\begin{align*}
		&\Pr\left[\sum_{i=1}^s X_i \geq s\tau\right]\\
		&= \Pr\left[X - \expect{X} \geq s\tau - \expect{X}\right]\\
		&\leq \exp\left[-\frac{\frac{\left(s\tau - \expect{X}\right)^2}{2}}{s \expect{X}\left(1-\expect{X}\right) + \frac{s\tau - \expect{X}}{3}}\right]\\
		&\leq \exp\left[-\frac{\left(s\tau - \expect{X}\right)^2}{2\left(s \expect{X} + \frac{s\tau - \expect{X}}{3}\right)}\right]\\
    &\leq \exp\left[-\frac{s^2\xi^2}{2\left(s
    \lambda + \frac{s\xi}{3}\right)}\right]\\
    &\leq \exp\left[ - \frac{s \xi^2}{2\left(\eps_1^2 t^2\left(1 +
    \frac{1}{10}t^2\right) + \frac{287}{80}\eps_1 \alpha t^2 +
    \frac{49}{1600}\eps_2 \alpha t^2 + \frac{\xi}{3}\right)}\right] && \paren*{\text{\cref{eq:lambda}}}\\
    &\leq \exp\left[ - \frac{s \xi^2}{\frac{1}{2}\left(\eps_2^2 \left(1 + \frac{1}{40}
    + \frac{287}{800} + \frac{49}{16000} + \frac{1}{6}\right) \right)}\right]
    && \paren*{\text{\cref{eq:xiub}, }\eps_1 < \eps_2,\, t \leq \frac{1}{2},\, \alpha \leq \frac{\eps_2}{100}}\\
    &\leq \exp\left[ - \frac{s }{78}
    \frac{(\eps_2^2-\eps_1^2)^3}{\eps_2^4}\right]
    &&\paren{\text{\cref{eq:xilb}}}\\
	\end{align*} 
  where the fourth line now follows due to the expression in the exponential
  being monotonically decreasing with respect to $\expect{X} \in [0, \tau)$.
	Therefore, setting
	\[
		s = 78\frac{\eps_2^4}{(\eps_2^2-\eps_1^2)^3}\ln(2/\delta)
	\]
  suffices for us to succeed at distinguishing the two cases with probability
  at least $1-\delta/2$.

  By \cref{lm:final_state}, \cref{alg:primitive} has at most a $\frac{99}{98}
  \alpha t < \frac{99}{19600} \frac{(\eps_2^2-\eps_1^2)^{3/2}}{\eps_2^2} \leq
  \frac{99}{19600}$ chance of failure. By applying
  \cref{corollary:bernoulli_bound_mult},
	\[
    s^\prime = \frac{2}{1-\frac{99}{19600}}\left(s + \ln(2/\delta)\right) \le
    \frac{2}{1-\frac{99}{19600}} s \cdot \ln(2/\delta) \leq  157
    \frac{\eps_2^4}{(\eps_2^2-\eps_1^2)^3}\ln(2/\delta) 
	\] suffices to achieve $s$ successful runs with
	probability $1 - \delta/2$.
	By the union bound, we will correctly differentiate the two
	cases with probability at least $1 - \delta$.
	
	The total time spent evolving under $H$ is then
	\begin{align*}
		s^\prime t &\leq 157 \frac{\eps_2^4}{(\eps_2^2-\eps_1^2)^3}\ln(2/\delta)  \cdot \frac{\sqrt{\eps_2^2-\eps_1^2}}{2\eps_2}\\
		&\leq 79\frac{\eps_2^3}{\left((\eps_2-\eps_1)(\eps_2+\eps_1)\right)^{5/2}} \ln(2/\delta)\\
		&\leq 79 \sqrt{\frac{\eps_2}{(\eps_2-\eps_1)^5}} \ln(2/\delta) = \bO{\sqrt{\frac{\eps_2}{(\eps_2-\eps_1)^5}} \log(1/\delta)},
	\end{align*}
	with a total number of queries equal to
	\begin{align*}
    s^\prime m &\leq 157 \frac{\eps_2^4}{\left(\eps_2^2-\eps_1^2\right)^3}
    \ln(2/\delta) \cdot \frac{50}{\sqrt{\eps_2^2-\eps_1^2}}\\
		&\leq 7850 \frac{\eps_2^4}{(\eps_2^2-\eps_1^2)^{7/2}} \ln(2/\delta)\\
		&\leq 7850 \sqrt{\frac{\eps_2}{(\eps_2-\eps_1)^7}} = \bO{\sqrt{\frac{\eps_2}{(\eps_2-\eps_1)^7}}}. && \qedhere
	\end{align*}
		
\end{proof}

\section{Lower Bound}

We will utilize the following fact about diamond distance of unitaries that will make calculations easier, at a loss of some constant factors.

\begin{fact}[{\cite[Proposition 1.6]{Haah2023query}}]\label{fact:phaseop-to-diamond}
	For all unitaries $U$ and $V$ of equal dimension, \[\frac{1}{2} \lVert U - V \rVert_\diamond \leq \min_{\theta \in [0, 2\pi)} \lVert e^{i \theta}U - V \rVert_\infty \leq \lVert U - V \rVert_\diamond.\]
\end{fact}

We now show our lower bound for $k$-locality testing, simply by showing that the statistical distance of the resulting unitaries (i.e., diamond distance) only grows linearly with time.

\begin{definition}
	For $0 \leq k \leq n$, we define \[Z_{1:k} \coloneqq \bigotimes_{i=1}^k Z \otimes  \bigotimes_{j=k+1}^n I\] to be the tensor product of $Z$ on the first $k$ qubits and identity on the last $n-k$ qubits.
\end{definition}

\begin{lemma}
	\label{lm:fractionalquery}
	For $0 \leq \eps_1 \leq \eps_2$
	\[
		\norm{e^{-i Z_{1:k} \eps_1 t} - e^{-i Z_{1:k} \eps_2 t} }_\diamond \leq 2(\eps_1 - \eps_2)t.
	\]
\end{lemma}
\begin{proof}
	Since $Z_{1:k}$ is diagonal with $\pm 1$ entries, $e^{-i Z_{1:k} \eps t}$ is diagonal with entries $e^{\mp i \eps t}$.
	Therefore, the eigenvalues of $e^{i \theta} \cdot e^{-i Z_{1:k} \eps_1 t} - e^{-i Z_{1:k} \eps_2 t}$ can be directly calculated, giving us
	\begin{align*}
		\min_{\theta \in [0, 2\pi)} \norm{e^{i \theta} \cdot e^{-i Z_{1:k} \eps_1 t} - e^{-i H \eps_2 t}}_\infty
			&= \min_{\theta \in [0, 2\pi)} \max\left(\abs{e^{i (\theta - \eps_1 t)} - e^{-i \eps_2 t}}, \abs{e^{i (\theta + \eps_1 t)} - e^{i \eps_2 t}}\right)\\
			&= \min\left( \abs{e^{-i \eps_1 t} - e^{-i \eps_2 t}}, \abs{e^{-i \eps_1 t} + e^{-i \eps_2 t}}\right)\\
			&= 2 \min \left(\left|\sin\left(\frac{(\eps_2 - \eps_1)t}{2}\right)\right|, \left|\cos\left(\frac{(\eps_2 - \eps_1)t}{2}\right)\right|\right)\\
			& \leq (\eps_2-\eps_1)t,
	\end{align*}
	where one of $\theta \in \{0, \pi\}$ minimizes the value via symmetry.
	By \cref{fact:phaseop-to-diamond}, $\norm{e^{-i Z_{1:k} \eps_1 t} - e^{-i Z_{1:k} \eps_2 t} }_\diamond \leq 2(\eps_1 - \eps_2)t$.\footnote{A direct calculation of the diamond distance will give an upper bound of $(\eps_2-\eps_1)t$, without the factor of $2$ from \cref{fact:phaseop-to-diamond}. See \cite[Proof of Proposition 1.6]{haah2024learning}.}
\end{proof}

\begin{remark}\label{remark:inverse}
	\cref{lm:fractionalquery} easily extends to the scenario where one is allowed to make calls to the inverse oracle, controlled versions of the oracle, the complex conjugate of the oracle, and any combination of these augmentations, as the diamond distance between the corresponding unitaries can be bounded as a function of time evolution.
\end{remark}

We are now ready to prove our tolerant locality testing lower bound by reducing to \cref{lm:fractionalquery}.
\lowerbound*
\begin{proof}
  Observe that for any $k^\prime > k$, $H_1 \coloneqq \eps_1 Z_{1:k^\prime}$ is within $\eps_1$ of being $k$-local and $H_2 \coloneqq \eps_2 Z_{1:k^\prime}$ is
  likewise $\eps_2$-far from being $k$-local.  $\norm{H_1}_\infty \leq
  \norm{H_2}_\infty \leq 1$ is also satisfied.  Let $t_i$ be the time
  evolution for each query in our algorithm.  By \cref{lm:fractionalquery}, the
  diamond distance between the time evolution of these two cases is at most
  $2(\eps_2 - \eps_1)t_i$ for each query. By the
  sub-additivity of diamond distance, a total time evolution of
  $\sum_i t_i = \bOm{(\eps_2 - \eps_1)^{-1}}$ is required to distinguish $H_1$ and $H_2$
  with constant probability.
\end{proof}

\begin{remark}
	\cref{thm:lower} also holds when the distance to $k$-locality is determined by operator norm $\norm{\cdot}_\infty$, any \emph{normalized} schatten $p$-norm $\norm{X}_p \coloneqq \frac{1}{2^{n/p}}\Tr\left(\abs{X}^p\right)^{\frac{1}{p}}$, or any Pauli decomposition $p$-norm $\norm{X}_{\text{Pauli}, p} \coloneqq \left(\sum_{P \in \pauli} \abs{\alpha_P}^p\right)^{\frac{1}{p}}$ for $X = \sum_{P \in \pauli} \alpha_P P$, improving upon that of \cite[Theorem 3.6]{bluhm2024hamiltonianpropertytesting}.
	This is simply  because the distance of $\eps Z_{1:k^\prime}$ (for $k^\prime > k$) from being $k$-local is exactly $\eps$ for all of these distance measures.
\end{remark}

\section*{Acknowledgements}
Sandia National Laboratories is a multimission laboratory managed and operated
by National Technology and Engineering Solutions of Sandia, LLC, a wholly
owned subsidiary of Honeywell International, Inc., for the U.S. Department of
Energy’s National Nuclear Security Administration under contract DE-NA-0003525.
This work was supported by the U.S. Department of Energy, Office of Science,
Office of Advanced Scientific Computing Research, Accelerated Research in
Quantum Computing, Fundamental Algorithmic Research for Quantum Utility, with
support also acknowledged from Fundamental Algorithm Research for Quantum
Computing.

This work was performed, in part, at the Center for Integrated
Nanotechnologies, an Office of Science User Facility operated for the U.S.
Department of Energy (DOE) Office of Science. Sandia National Laboratories is a
multimission laboratory managed and operated by National Technology \&
Engineering Solutions of Sandia, LLC, a wholly owned subsidiary of Honeywell
International, Inc., for the U.S. DOE’s National Nuclear Security
Administration under contract DE-NA-0003525. The views expressed in the article
do not necessarily represent the views of the U.S. DOE or the United States
Government.

DL is supported by US NSF Award CCF-222413.

We would like to thank Vishnu Iyer and Justin Yirka for getting us started on this
problem.  We would also like to thank Francisco Escudero Gutiérrez, Srinivasan
Arunachalam, and Fang Song for insightful feedback and comments.

\newpage
\bibliographystyle{alphaurl}
\bibliography{refs}

\newcommand{\etalchar}[1]{$^{#1}$}
\begin{thebibliography}{BLMT24b}

\bibitem[AAKS21]{anshu2021sample}
Anurag Anshu, Srinivasan Arunachalam, Tomotaka Kuwahar, and Mehdi Soleimanifar.
\newblock Sample-efficient learning of interacting quantum systems.
\newblock {\em Nature Physics}, 17:931--935, Aug 2021.
\newblock \href {https://doi.org/10.1038/s41567-021-01232-0}
  {\path{doi:10.1038/s41567-021-01232-0}}.

\bibitem[AAV13]{aharonov2013quantumpcpconjecture}
Dorit Aharonov, Itai Arad, and Thomas Vidick.
\newblock The quantum {PCP} conjecture, 2013.
\newblock \href {http://arxiv.org/abs/1309.7495} {\path{arXiv:1309.7495}}.

\bibitem[ADG24]{arunachalam2024testinglearningstructuredquantum}
Srinivasan Arunachalam, Arkopal Dutt, and Francisco~Escudero Gutiérrez.
\newblock Testing and learning structured quantum hamiltonians, 2024.
\newblock \href {http://arxiv.org/abs/2411.00082} {\path{arXiv:2411.00082}}.

\bibitem[BCO24]{bluhm2024hamiltonianpropertytesting}
Andreas Bluhm, Matthias~C. Caro, and Aadil Oufkir.
\newblock {Hamiltonian Property Testing}, 2024.
\newblock \href {http://arxiv.org/abs/2403.02968} {\path{arXiv:2403.02968}}.

\bibitem[BHMT02]{Brassard_2002}
Gilles Brassard, Peter H{\o}yer, Michele Mosca, and Alain Tapp.
\newblock {Quantum Amplitude Amplification and Estimation}, 2002.
\newblock \href {https://doi.org/10.1090/conm/305/05215}
  {\path{doi:10.1090/conm/305/05215}}.

\bibitem[BLMT24a]{bakshi2024learning}
Ainesh Bakshi, Allen Liu, Ankur Moitra, and Ewin Tang.
\newblock Learning quantum hamiltonians at any temperature in polynomial time.
\newblock In {\em Proceedings of the 56th Annual ACM Symposium on Theory of
  Computing}, STOC 2024, page 1470–1477, New York, NY, USA, 2024. Association
  for Computing Machinery.
\newblock \href {https://doi.org/10.1145/3618260.3649619}
  {\path{doi:10.1145/3618260.3649619}}.

\bibitem[BLMT24b]{bakshi2024structure}
Ainesh Bakshi, Allen Liu, Ankur Moitra, and Ewin Tang.
\newblock Structure learning of hamiltonians from real-time evolution.
\newblock In {\em 2024 IEEE 65th Annual Symposium on Foundations of Computer
  Science (FOCS)}, pages 1037--1050, 2024.
\newblock \href {https://doi.org/10.1109/FOCS61266.2024.00069}
  {\path{doi:10.1109/FOCS61266.2024.00069}}.

\bibitem[CMN{\etalchar{+}}18]{childs2018simulation}
Andrew~M. Childs, Dmitri Maslov, Yunseong Nam, Neil~J. Ross, and Yuan Su.
\newblock Toward the first quantum simulation with quantum speedup.
\newblock {\em Proceedings of the National Academy of Sciences},
  115(38):9456--9461, 2018.
\newblock \href {https://doi.org/10.1073/pnas.1801723115}
  {\path{doi:10.1073/pnas.1801723115}}.

\bibitem[EV93]{elitzur1993bomb}
Avshalom~C. Elitzur and Lev Vaidman.
\newblock Quantum mechanical interaction-free measurements.
\newblock {\em Foundations of Physics}, 23:987--997, 1993.
\newblock \href {https://doi.org/10.1007/BF00736012}
  {\path{doi:10.1007/BF00736012}}.

\bibitem[FP08]{Facchi_2008}
P~Facchi and S~Pascazio.
\newblock Quantum zeno dynamics: mathematical and physical aspects.
\newblock {\em Journal of Physics A: Mathematical and Theoretical},
  41(49):493001, oct 2008.
\newblock \href {https://doi.org/10.1088/1751-8113/41/49/493001}
  {\path{doi:10.1088/1751-8113/41/49/493001}}.

\bibitem[GIKL23]{grewal_et_al:LIPIcs.ITCS.2023.64}
Sabee Grewal, Vishnu Iyer, William Kretschmer, and Daniel Liang.
\newblock {Low-Stabilizer-Complexity Quantum States Are Not Pseudorandom}.
\newblock In Yael Tauman~Kalai, editor, {\em 14th Innovations in Theoretical
  Computer Science Conference (ITCS 2023)}, volume 251 of {\em Leibniz
  International Proceedings in Informatics (LIPIcs)}, pages 64:1--64:20,
  Dagstuhl, Germany, 2023. Schloss Dagstuhl -- Leibniz-Zentrum f{\"u}r
  Informatik.
\newblock \href {https://doi.org/10.4230/LIPIcs.ITCS.2023.64}
  {\path{doi:10.4230/LIPIcs.ITCS.2023.64}}.

\bibitem[Gut24]{gutierrez2024simplealgorithmstestlearn}
Francisco~Escudero Gutiérrez.
\newblock {Simple algorithms to test and learn local Hamiltonians}, 2024.
\newblock \href {http://arxiv.org/abs/2404.06282} {\path{arXiv:2404.06282}}.

\bibitem[HKOT23]{Haah2023query}
J.~Haah, R.~Kothari, R.~O'Donnell, and E.~Tang.
\newblock Query-optimal estimation of unitary channels in diamond distance.
\newblock In {\em 2023 IEEE 64th Annual Symposium on Foundations of Computer
  Science (FOCS)}, pages 363--390, Los Alamitos, CA, USA, nov 2023. IEEE
  Computer Society.
\newblock \href {https://doi.org/10.1109/FOCS57990.2023.00028}
  {\path{doi:10.1109/FOCS57990.2023.00028}}.

\bibitem[HKT24]{haah2024learning}
Jeongwan Haah, Robin Kothari, and Ewin Tang.
\newblock Learning quantum hamiltonians from high-temperature gibbs states and
  real-time evolutions.
\newblock {\em Nature Physics}, 20:1027--1031, june 2024.
\newblock \href {https://doi.org/10.1038/s41567-023-02376-x}
  {\path{doi:10.1038/s41567-023-02376-x}}.

\bibitem[HTFS23]{huang2023many}
Hsin-Yuan Huang, Yu~Tong, Di~Fang, and Yuan Su.
\newblock Learning many-body hamiltonians with heisenberg-limited scaling.
\newblock {\em Phys. Rev. Lett.}, 130:200403, May 2023.
\newblock \href {https://doi.org/10.1103/PhysRevLett.130.200403}
  {\path{doi:10.1103/PhysRevLett.130.200403}}.

\bibitem[Llo96]{lloyd1996universal}
Seth Lloyd.
\newblock Universal quantum simulators.
\newblock {\em Science}, 273(5278):1073--1078, 1996.
\newblock \href {https://doi.org/10.1126/science.273.5278.1073}
  {\path{doi:10.1126/science.273.5278.1073}}.

\bibitem[MO10]{montanaro2010quantumbooleanfunctions}
Ashley Montanaro and Tobias~J. Osborne.
\newblock Quantum boolean functions, 2010.
\newblock URL: \url{https://arxiv.org/abs/0810.2435}, \href
  {http://arxiv.org/abs/0810.2435} {\path{arXiv:0810.2435}}.

\bibitem[MW16]{MdW13}
Ashley Montanaro and Ronald~{de} Wolf.
\newblock {\em A Survey of Quantum Property Testing}.
\newblock Number~7 in Graduate Surveys. Theory of Computing Library, 2016.
\newblock \href {https://doi.org/10.4086/toc.gs.2016.007}
  {\path{doi:10.4086/toc.gs.2016.007}}.

\bibitem[SY23]{she_et_al:LIPIcs.ITCS.2023.96}
Adrian She and Henry Yuen.
\newblock {Unitary Property Testing Lower Bounds by Polynomials}.
\newblock In Yael Tauman~Kalai, editor, {\em 14th Innovations in Theoretical
  Computer Science Conference (ITCS 2023)}, volume 251 of {\em Leibniz
  International Proceedings in Informatics (LIPIcs)}, pages 96:1--96:17,
  Dagstuhl, Germany, 2023. Schloss Dagstuhl -- Leibniz-Zentrum f{\"u}r
  Informatik.
\newblock \href {https://doi.org/10.4230/LIPIcs.ITCS.2023.96}
  {\path{doi:10.4230/LIPIcs.ITCS.2023.96}}.

\bibitem[VO21]{venkateswaran_et_al:LIPIcs.STACS.2021.59}
Ramgopal Venkateswaran and Ryan O'Donnell.
\newblock {Quantum Approximate Counting with Nonadaptive Grover Iterations}.
\newblock In Markus Bl\"{a}ser and Benjamin Monmege, editors, {\em 38th
  International Symposium on Theoretical Aspects of Computer Science (STACS
  2021)}, volume 187 of {\em Leibniz International Proceedings in Informatics
  (LIPIcs)}, pages 59:1--59:12, Dagstuhl, Germany, 2021. Schloss Dagstuhl --
  Leibniz-Zentrum f{\"u}r Informatik.
\newblock \href {https://doi.org/10.4230/LIPIcs.STACS.2021.59}
  {\path{doi:10.4230/LIPIcs.STACS.2021.59}}.

\bibitem[Wil17]{wilde2017quantum}
Mark~M Wilde.
\newblock {\em Quantum Information Theory}.
\newblock Cambridge university press, 2 edition, 2017.

\end{thebibliography}

\newpage
\appendix
\section{Optimal Tolerant Testing with Inverse Queries}
In this section we augment the tolerant testing algorithm in
\cite{gutierrez2024simplealgorithmstestlearn,arunachalam2024testinglearningstructuredquantum},
with amplitude estimation to get an optimal tolerant tester when given access
to controlled versions of the forward and reverse time
evolution.\footnote{Using the multiplicative error form from
\cite{venkateswaran_et_al:LIPIcs.STACS.2021.59} should allow for one to remove
the need for controlled access while remaining non-adaptive, though it causes
the constants to blow-up.}

We begin with the following crucial result of Gutiérrez.

\begin{lemma}[{\cite[Lemma 3.1]{arunachalam2024testinglearningstructuredquantum}}]\label{lem:gutierrez-upper}
  Let $0 \leq \eps_1 \leq \eps_2 \leq 1$. Let $\alpha \coloneqq \frac{\eps_2 -
  \eps_1}{3c}$ and $H$ be an $n$-qubit Hamiltonian with $\lVert H \rVert_\infty
  = 1$.  Define $U \coloneqq e^{-i H \alpha}$, and let $U_{> k}$ be
  $U\ket{\maxentangled}$ projected onto onto the space spanned by $\set{(I
  \otimes P)\ket{\maxentangled} : P \in \set{I, X, Y, Z}^{\otimes n},
  \abs{P} > k}$.  We have that if $H$ is $\eps_1$-close to being $k$-local,
  then 
	\[
	\norm{U_{> k}}_2^2 \leq \left((\eps_2- \eps_1)\frac{2\eps_1 + \eps_2}{9c}\right)^2,
	\]
	and if $H$ is $\eps_2$-far from being $k$-local, then
	\[
	\norm{U_{> k}}_2^2 \geq \left((\eps_2- \eps_1)\frac{\eps_1 + 2\eps_2}{9c}\right)^2.
	\]
\end{lemma}

We also cite the following result of \cite{grewal_et_al:LIPIcs.ITCS.2023.64}, which itself follows as a corollary of the celebrated Quantum Amplitude Estimation \cite[Theorem 12]{Brassard_2002} result.
\begin{lemma}[Quantum Amplitude Estimation {\cite[Corollary 29]{grewal_et_al:LIPIcs.ITCS.2023.64}}]\label{thm:qae}
	Let $\Pi$ be a projector and $\ket{\psi}$ be an $n$-qubit pure state such that $\braket{\psi}{\Pi|\psi} = \eta$.
	Given access to the unitary transformations $R_{\Pi} = 2\Pi - I$ and $R_{\psi} = 2 \ket\psi\!\!\bra\psi - I$, there exists a quantum algorithm that outputs $\wh{\eta}$ such that 
	\[\abs{\wh{\eta}- \eta} \leq \xi\]
	with probability at least $\frac{8}{\pi^2}$. The algorithm makes no more than $\pi \frac{\sqrt{\eta(1-\eta)+\xi}}{\xi}$ calls to the controlled versions of $R_\Pi$ and $R_\psi$.
\end{lemma}
In particular, this implies that if we have (controlled) query access to $U$,
$U^*$ for some unitary $U$, and a known state $\ket{\phi}$, we can estimate
$\eta = \norm{\Pi U\ket{\phi}}_2^2$ to $\zeta$ accuracy by defining $\ket{\psi}
\coloneqq U\ket{\phi}$ and implementing $R_\psi$ with controlled applications of $U$.

We are now ready to state the algorithm, which can be seen as the algorithm of
\cite{gutierrez2024simplealgorithmstestlearn,arunachalam2024testinglearningstructuredquantum} augmented with \cref{thm:qae}.

\ttinv*
\begin{proof}
Let $U \coloneqq e^{-i H \alpha}$ as in \cref{lem:gutierrez-upper}. We apply
\cref{lem:gutierrez-upper} with $\Pi$ the projector onto the space spanned by
$\set{(I \otimes P)\ket{\maxentangled} : P \in \set{I, X, Y, Z}^{\otimes
n}, \abs{P} > k}$ to estimate  $\norm{U_{> k}}_2^2$. Observe that the
absolute difference between the two terms in \cref{lem:gutierrez-upper} is
\[
\left((\eps_2- \eps_1)\frac{\eps_1 + 2\eps_2}{9c}\right)^2 - \left((\eps_2- \eps_1)\frac{2\eps_1 + \eps_2}{9c}\right)^2 = \frac{(\eps_2-\eps_1)^3(\eps_2+\eps_1)}{27c^2}.
\]
Therefore, we can distinguish the two cases to constant success probability by
estimating $\eta = \norm{U_{> k}}_2^2$ to error $\zeta =
\frac{(\eps_2-\eps_1)^3(\eps_2+\eps_1)}{54c^2}$.  By \cref{thm:qae}, the number
of queries is then no more than
\begin{align*}
	&\pi \frac{\sqrt{(\eps_2-\eps_1)^2(\eps_1 + 2 \eps_2)^2/(81c^2) + (\eps_2-\eps_1)^3(\eps_1+\eps_2)/(54c^2)}}{(\eps_2-\eps_1)^3(\eps_1+\eps_2)/(54c^2)}\\
	=& \frac{54 \pi c}{(\eps_2-\eps_1)^2} \frac{\sqrt{(\eps_1 + 2 \eps_2)^2/81 + (2\eps_2-2\eps_1)(2\eps_1+2\eps_2)/216}}{\eps_1+\eps_2}\\
	\leq& \frac{54 \pi c}{(\eps_2-\eps_1)^2} \frac{\sqrt{(2\eps_1 + 2 \eps_2)^2/81 + (2\eps_1+2\eps_2)^2/216}}{\eps_1+\eps_2}\\
	\leq& \frac{54 \pi c}{(\eps_2-\eps_1)^2} \frac{\sqrt{11(2\eps_1 + 2 \eps_2)^2/648 }}{\eps_1+\eps_2}\\
	\leq&  \frac{3 \sqrt{22} \pi c}{(\eps_2 - \eps_1)^2}.
\end{align*}

Since the Hamiltonian is applied for $\alpha \coloneqq \frac{\eps_2 - \eps_1}{3c}$ for each query, the total evolution of the Hamiltonian is at most
\[
\frac{3 \sqrt{22} \pi c}{(\eps_2 - \eps_1)^2}\frac{\eps_2 - \eps_1}{3c} = \frac{\sqrt{22} \pi }{\eps_2-\eps_1}.
\]
By standard error reduction, we can reduce the constant failure probability to at most $\delta$ using $\log(1/\delta)$ repetitions.

Finally, observe that constructing $R_\Pi$ (and its controlled version), as in \cref{thm:qae} is free, as $\Pi$ is a known projector onto the low locality Paulis.
On the other hand, $R_\psi$ requires us to take (a version of) the Grover Diffusion operator $D \coloneqq 2\ket{0}\!\bra{0} - I$ and conjugate it by $U$.
This is the step that requires access to $U^\dagger \coloneqq e^{i H \alpha}$.

\end{proof}

Since this matches the lower bound of \cref{thm:lower}, \cref{thm:tight-tolerant-inverse} is optimal.

\end{document}